\documentclass[pra,aps,showpacs,onecolumn,twoside,superscriptaddress]{revtex4}

\usepackage{amsmath,amsfonts,amssymb,caption,color,epsfig,graphics,graphicx,hyperref,latexsym,mathrsfs,revsymb,theorem,url,verbatim,epstopdf}

\hypersetup{colorlinks,linkcolor={blue},citecolor={red},urlcolor={blue}}

\newtheorem{definition}{Definition}
\newtheorem{proposition}[definition]{Proposition}
\newtheorem{lemma}[definition]{Lemma}

\newtheorem{theorem}[definition]{Theorem}
\newtheorem{corollary}[definition]{Corollary}
\newtheorem{conjecture}[definition]{Conjecture}

\newtheorem{remark}[definition]{Remark}
\newtheorem{example}[definition]{Example}
\newtheorem{question}[definition]{Question}

\def\bcj{\begin{conjecture}}
\def\ecj{\end{conjecture}}
\def\bcr{\begin{corollary}}
\def\ecr{\end{corollary}}
\def\bd{\begin{definition}}
\def\ed{\end{definition}}
\def\bea{\begin{eqnarray}}
\def\eea{\end{eqnarray}}
\def\bem{\begin{enumerate}}
\def\eem{\end{enumerate}}
\def\bex{\begin{example}}
\def\eex{\end{example}}
\def\bim{\begin{itemize}}
\def\eim{\end{itemize}}
\def\bl{\begin{lemma}}
\def\el{\end{lemma}}
\def\bma{\begin{bmatrix}}
\def\ema{\end{bmatrix}}
\def\bpf{\begin{proof}}
\def\epf{\end{proof}}
\def\bpp{\begin{proposition}}
\def\epp{\end{proposition}}
\def\bqu{\begin{question}}
\def\equ{\end{question}}
\def\br{\begin{remark}}
\def\er{\end{remark}}
\def\bt{\begin{theorem}}
\def\et{\end{theorem}}

\def\squareforqed{\hbox{\rlap{$\sqcap$}$\sqcup$}}
\def\qed{\ifmmode\squareforqed\else{\unskip\nobreak\hfil
\penalty50\hskip1em\null\nobreak\hfil\squareforqed
\parfillskip=0pt\finalhyphendemerits=0\endgraf}\fi}
\def\endenv{\ifmmode\;\else{\unskip\nobreak\hfil
\penalty50\hskip1em\null\nobreak\hfil\;
\parfillskip=0pt\finalhyphendemerits=0\endgraf}\fi}
\newenvironment{proof}{\noindent \textbf{{Proof.~} }}{\qed}
\def\Dbar{\leavevmode\lower.6ex\hbox to 0pt
{\hskip-.23ex\accent"16\hss}D}
\makeatletter
\def\url@leostyle{%
  \@ifundefined{selectfont}{\def\UrlFont{\sf}}{\def\UrlFont{\small\ttfamily}}}
\makeatother
\def\bcj{\begin{conjecture}}
\def\ecj{\end{conjecture}}
\def\bcr{\begin{corollary}}
\def\ecr{\end{corollary}}
\def\bd{\begin{definition}}
\def\ed{\end{definition}}
\def\bea{\begin{eqnarray}}
\def\eea{\end{eqnarray}}
\def\bem{\begin{enumerate}}
\def\eem{\end{enumerate}}
\def\bex{\begin{example}}
\def\eex{\end{example}}
\def\bim{\begin{itemize}}
\def\eim{\end{itemize}}
\def\bl{\begin{lemma}}
\def\el{\end{lemma}}
\def\bpf{\begin{proof}}
\def\epf{\end{proof}}
\def\bpp{\begin{proposition}}
\def\epp{\end{proposition}}
\def\bqu{\begin{question}}
\def\equ{\end{question}}
\def\br{\begin{remark}}
\def\er{\end{remark}}
\def\bt{\begin{theorem}}
\def\et{\end{theorem}}

\def\btb{\begin{tabular}}
\def\etb{\end{tabular}}

\newcommand{\nc}{\newcommand}

\def\a{\alpha}
\def\b{\beta}
\def\g{\gamma}
\def\d{\delta}
\def\e{\epsilon}

\def\z{\zeta}

\def\t{\theta}

\def\l{\lambda}
\def\m{\mu}
\def\n{\nu}
\def\x{\xi}
\def\p{\pi}
\def\r{\rho}
\def\s{\sigma}

\def\ph{\varphi}

\def\ps{\psi}
\def\o{\omega}

\def\G{\Gamma}

 \nc{\bbA}{\mathbb{A}} \nc{\bbB}{\mathbb{B}} \nc{\bbC}{\mathbb{C}}
 \nc{\bbD}{\mathbb{D}} \nc{\bbE}{\mathbb{E}} \nc{\bbF}{\mathbb{F}}
 \nc{\bbG}{\mathbb{G}} \nc{\bbH}{\mathbb{H}} \nc{\bbI}{\mathbb{I}}
 \nc{\bbJ}{\mathbb{J}} \nc{\bbK}{\mathbb{K}} \nc{\bbL}{\mathbb{L}}
 \nc{\bbM}{\mathbb{M}} \nc{\bbN}{\mathbb{N}} \nc{\bbO}{\mathbb{O}}
 \nc{\bbP}{\mathbb{P}} \nc{\bbQ}{\mathbb{Q}} \nc{\bbR}{\mathbb{R}}
 \nc{\bbS}{\mathbb{S}} \nc{\bbT}{\mathbb{T}} \nc{\bbU}{\mathbb{U}}
 \nc{\bbV}{\mathbb{V}} \nc{\bbW}{\mathbb{W}} \nc{\bbX}{\mathbb{X}}
 \nc{\bbZ}{\mathbb{Z}}

 \nc{\bA}{{\bf A}} \nc{\bB}{{\bf B}} \nc{\bC}{{\bf C}}
 \nc{\bD}{{\bf D}} \nc{\bE}{{\bf E}} \nc{\bF}{{\bf F}}
 \nc{\bG}{{\bf G}} \nc{\bH}{{\bf H}} \nc{\bI}{{\bf I}}
 \nc{\bJ}{{\bf J}} \nc{\bK}{{\bf K}} \nc{\bL}{{\bf L}}
 \nc{\bM}{{\bf M}} \nc{\bN}{{\bf N}} \nc{\bO}{{\bf O}}
 \nc{\bP}{{\bf P}} \nc{\bQ}{{\bf Q}} \nc{\bR}{{\bf R}}
 \nc{\bS}{{\bf S}} \nc{\bT}{{\bf T}} \nc{\bU}{{\bf U}}
 \nc{\bV}{{\bf V}} \nc{\bW}{{\bf W}} \nc{\bX}{{\bf X}}
 \nc{\bZ}{{\bf Z}}

\nc{\cA}{{\cal A}} \nc{\cB}{{\cal B}} \nc{\cC}{{\cal C}}
\nc{\cD}{{\cal D}} \nc{\cE}{{\cal E}} \nc{\cF}{{\cal F}}
\nc{\cG}{{\cal G}} \nc{\cH}{{\cal H}} \nc{\cI}{{\cal I}}
\nc{\cJ}{{\cal J}} \nc{\cK}{{\cal K}} \nc{\cL}{{\cal L}}
\nc{\cM}{{\cal M}} \nc{\cN}{{\cal N}} \nc{\cO}{{\cal O}}
\nc{\cP}{{\cal P}} \nc{\cQ}{{\cal Q}} \nc{\cR}{{\cal R}}
\nc{\cS}{{\cal S}} \nc{\cT}{{\cal T}} \nc{\cU}{{\cal U}}
\nc{\cV}{{\cal V}} \nc{\cW}{{\cal W}} \nc{\cX}{{\cal X}}
\nc{\cZ}{{\cal Z}}

\nc{\hA}{{\hat{A}}} \nc{\hB}{{\hat{B}}} \nc{\hC}{{\hat{C}}}
\nc{\hD}{{\hat{D}}} \nc{\hE}{{\hat{E}}} \nc{\hF}{{\hat{F}}}
\nc{\hG}{{\hat{G}}} \nc{\hH}{{\hat{H}}} \nc{\hI}{{\hat{I}}}
\nc{\hJ}{{\hat{J}}} \nc{\hK}{{\hat{K}}} \nc{\hL}{{\hat{L}}}
\nc{\hM}{{\hat{M}}} \nc{\hN}{{\hat{N}}} \nc{\hO}{{\hat{O}}}
\nc{\hP}{{\hat{P}}} \nc{\hR}{{\hat{R}}} \nc{\hS}{{\hat{S}}}
\nc{\hT}{{\hat{T}}} \nc{\hU}{{\hat{U}}} \nc{\hV}{{\hat{V}}}
\nc{\hW}{{\hat{W}}} \nc{\hX}{{\hat{X}}} \nc{\hZ}{{\hat{Z}}}

\nc{\hn}{{\hat{n}}}

\def\dim{\mathop{\rm Dim}}

\def\lin{\mathop{\rm span}}

\def\min{\mathop{\rm min}}

\def\rank{\mathop{\rm rank}}

\def\tr{\mathop{\rm Tr}}

\def\dg{\dagger}

\def\ox{\otimes}

\def\ra{\rightarrow}

\newcommand{\bra}[1]{\langle#1|}
\newcommand{\ket}[1]{|#1\rangle}
\newcommand{\proj}[1]{| #1\rangle\!\langle #1 |}
\newcommand{\ketbra}[2]{|#1\rangle\!\langle#2|}

\def\Dbar{\leavevmode\lower.6ex\hbox to 0pt
{\hskip-.23ex\accent"16\hss}D}

\begin{document}
\title{Schmidt number of bipartite and multipartite states under local projections}

\title{Tripartite genuinely entangled states from entanglement-breaking subspaces}

\date{\today}

\pacs{03.65.Ud, 03.67.Mn}

\author{Yize Sun}\email[]{sunyize@buaa.edu.cn}
\affiliation{School of Mathematical Sciences, Beihang University, Beijing 100191, China}

\author{Lin Chen}\email[]{linchen@buaa.edu.cn (corresponding author)}
\affiliation{School of Mathematical Sciences, Beihang University, Beijing 100191, China}
\affiliation{International Research Institute for Multidisciplinary Science, Beihang University, Beijing 100191, China}

\begin{abstract}
The determination of genuine  entanglement is a central problem in quantum information processing. We investigate the tripartite state as the tensor product of two bipartite entangled states by merging two systems. We show that the tripartite state is a genuinely entangled state when the range of both bipartite states are entanglement-breaking subspaces. We further investigate the tripartite state when one of the two bipartite states has rank two. Our results provide the latest progress on a conjecture proposed in the paper [Yi Shen \textit{et al}, J. Phys. A 53, 125302 (2020)]. We apply our results to construct multipartite states whose bipartite reduced density operators have additive EOF. Further, such states are distillable across every bipartition under local operations and classical communications. 
\end{abstract}

\Large

\maketitle


\section{Introduction}
\label{sec:int}

Various quantum-information tasks requires genuine entanglement as an indispensable ingredient. They include the quantum key distribution \cite{das2019universal}, measurement-based quantum computation \cite{Briegel2009Measurement,RaussendorfA}, communication \cite{DeQuantum} and one-dimensional cluster-Ising model \cite{Giampaolo_2014}. The detection and construction of genuine entanglement has a received extensive attentions in quantum information both theoretically and experimentally in recent years. They have proposed methods via entanglement witnesses  \cite{sun2019improved,HuberDetection,HuberWitnessing,DeMultipartite,Huber2013Entropy,SperlingMultipartite,JungnitschTaming,CoffmanDistributed}, generalized concurrence  \cite{MaMeasure,ChenImproved,Hong2012Measure,Gao2014On}, Bell inequalities \cite{BancalDevice}, geometric measure \cite{roy2019computable} and semidefinite programming \cite{Lancien_2015}. Nevertheless, it is still generally hard to determine whether the multipartite state is a genuinely entangled (GE) state or not. The latter is also known as the biseparable state. Recently, constructing tripartite GE states via the tensor product of two bipartite entangled states has been proposed \cite{Shen_2020}, see also Conjecture \ref{cj:aotimesb}. The method offers the following advantage. The construction and detection of bipartite entanglement is a more operational task than that of multipartite entanglement \cite{HuberDetection}, and many methods such as the entanglement witnesses have been developed in the past decades \cite{VanMultipartite,WuQuantum}. Hence, constructing tripartite GE states using bipartite entangled states is a convenient and efficient method. In this paper, we investigate the following conjecture on genuine entanglement.
\begin{conjecture}
\label{cj:aotimesb}
If $\a_{AC_1},\b_{BC_2}$ are both bipartite entangled states, then $\a_{AC_1}\otimes\b_{BC_2}$ is a tripartite GE state on the Hilbert space $\cH_{ABC}$ with $C=C_1C_2$.
\end{conjecture}

In the conjecture, we may assume that a tensor product $\r_{ABC}=\a_{AC_1}\otimes\b_{BC_2}$ on the Hilbert space $\mathcal{H}_{ABC}$  
with a decomposition of system $C$ as $\mathcal{H}_{C}=\mathcal{H}_{C_1}\otimes\mathcal{H}_{C_2}$.

We have shown that the counterexample to Conjecture \ref{cj:aotimesb} might exist only if the range of $\a_{AC_1}$ and $\b_{BC_2}$ are both spanned by product vectors \cite{Shen_2020}. As far as we know, the relation between GE states and entanglement-breaking (EB) space is little studied yet. The EB space \cite{PhysRevLett.89.027901,Zhao_2019} is a bipartite subspace that offers advantage for the additivity of entanglement of formation (EOF) \cite{Zhao_2019}.  In the past years, the EOF has many applications  such as the measure of entanglement for two-qubit pure states, many-body systems and quantifying the amount of entanglement in \cite{PhysRevA.54.3824,PhysRevLett.85.2625, PhysRevLett.95.210501,PhysRevLett.121.190503,Zhu2010Additivity,0034-4885-81-7-074002, 1751-8121-46-39-395302, PhysRevLett.91.107901}. 
The aim of constructing EB spaces is to connect the entanglement cost and EOF. The entanglement cost is a physically motivated entanglement measure quantifying the least entanglement required to form a bipartite state asymptotically \cite{RevModPhys.81.865}.  Because the entanglement cost is exactly the regularized form of EOF \cite{Hayden2000The}, it is notoriously difficult to derivate the entanglement cost of a bipartite state unless the state has additive EOF.  Actually, if a state has additive EOF, then it makes the equality between the entanglement cost and EOF of the state. However,  the relation between additive EOF and GE states is not well studied.



In this paper we consider the case that the product vectors form an EB space namely $\ket{a_1,1},...,\ket{a_n,n}$ up to equivalence under stochastic local operations and classical communications (SLOCC). This is presented in Theorem \ref{thm:main}. Then we investigate the tensor product of a rank-two bipartite state $\a$ whose range is the simplest entanglement-breaking subspace, and an arbitrary bipartite state $\b$. The range of $\a$ is spanned by $\ket{0,0}$ and $\ket{1,1}$ up to equivalence, as we show in Lemma \ref{le:2r2oxr2}. In Lemma \ref{le:rank2}, we characterize the properties of $\a$ and $\b$ by assuming that the tensor product is a tripartite biseparable state.
Next, we apply our results to two types of constructing multipartite GE states. The first family of multipartite GE states have bipartite reduced density operators with additive EOF. The second family of multipartite GE states whose every bipartition produces a distillable state under LOCC. Such multipartite states are of widely usefulness in theory and experiment for various quantum-information tasks. Then by extending Conjecture \ref{cj:aotimesb}, we explore more ways to construct multipartite GE states. Moreover, we reveal connections between Conjecture \ref{cj:aotimesb} and Conjecture \ref{cj:aotimesbotimesc=sep} in Lemma \ref{thm:bi->tri}. 

The rest of this paper is organized as follows. In Sec. \ref{sec:pr} we introduce the preliminary facts and notations used in this paper.  In Sec.
\ref{sec:ebOXeb} we investigate the tensor product of two entanglement-breaking subspaces. We further study the tensor product of the simplest entanglement-breaking subspace and an arbitrary bipartite subspace. We apply our results in Sec. \ref{sec:app}. Furthermore we discuss our results in the multipartite space in Sec. \ref{sec:multiGE}. Finally we conclude in Sec. \ref{sec:con}.

\section{Preliminaries}
\label{sec:pr}

In this section we introduce the notions and facts used in this paper. Suppose $\rho_{A_1A_2\cdots A_n}$ is an $n$-partite state on the Hilbert space $\cH_{A_1A_2\cdots A_n}:=\cH_{A_1}\ox\cH_{A_2}\ox\cdots\ox\cH_{A_n}$. Denote $\rho_{A_1A_2\cdots A_n}$ by $\rho$ for simplicity, and denote by $\rho_{A_{j_1}A_{j_2}\cdots A_{j_k}}$ the reduced density matrices of $\rho$ of system $A_{j_1}A_{j_2}\cdots A_{j_k}$. Unless stated otherwise, we shall not normalize quantum states for convenience. For $\rho=\sum_{j=1}^k\proj{\psi_j}$, we denote by $\cR(\rho)$ the range of $\rho$. So we have $\cR(\rho)=\lin\{\ket{\psi_j}\}_{j=1}^k$. Next, we say that two $n$-partite states  $\a$ and $\b$ are locally equivalent when there exists a product invertible operation $X=X_1\ox...\ox X_n$
such that $\a=X\b X^\dg $. In this case, we also say that $\a$ and $\b$ are equivalent under stochastic local operations and classical communications (SLOCC) \cite{DThree}. Physically, $\a$ and $\b$ can be converted each other with some nonzero probability. In particular the probability reaches one when $X$ is unitary. 

In the following we review the entanglement-breaking (EB) space proposed in the paper \cite{PhysRevLett.89.027901} and recently studied in \cite{Zhao_2019}. Suppose that $A, B, a, b$ are four quantum systems, and $V$ a bipartite subspace of $\mathcal{H}_{AB}$ with the following property. Given
a bipartite pure state $\ket{\psi}\in V\otimes\mathcal{H}_{ab}$, tracing out system $B$ destroys the entanglement between $AB$ with $ab$. In other word the bipartite state of system $A$ and $ab$ is separable. Then we refer to such $V$ as an EB space, see Figure \ref{fig:eb}. Next, we introduce a fact from \cite{PhysRevLett.89.027901} as Lemma \ref{le:eb=additive=eof}.

\begin{figure}[htb]
	\includegraphics[scale=0.3,angle=0]{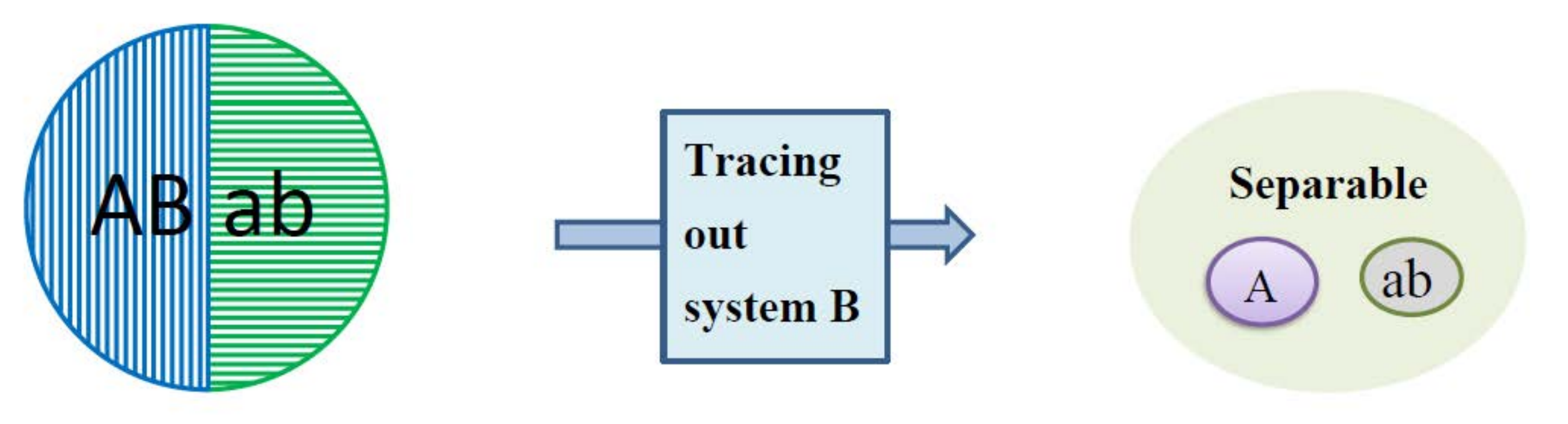}
	\caption{The description of an entanglement-breaking (EB) space.}
	\label{fig:eb}
\end{figure}

\begin{lemma}
	\label{le:eb=additive=eof}	
	The bipartite state whose range is an EB space has additive entanglement of formation (EOF). 
\end{lemma}

Recall that the EOF of a bipartite mixed state $\r$ is defined as 
$
E_f(\r):=\min_{\r=\sum_ip_i\proj{\psi_i}}\sum_ip_iE(\ket{\psi_i})
$. 
Physically, the EOF of $\r$ means the minimum entanglement
required for realizing the single copy of $\r$. The EOF of any state in an  optimal decomposition family is computable \cite{VollbrechtEntanglement}. In addition, the optimal decomposition is closely connected the additivity of EOF of $\r$  \cite{Chen2007Entanglement}. From the definition of additivity of EOF, it holds that 
$E_f(\r\otimes\s)=E_f(\r)+E_f(\s)$ for any bipartite state $\s$. Then we introduce a fact from \cite{Chen2007Entanglement} as follows.

\begin{lemma}
	\label{le:eof}
	If the EOF of two states are additive, then the EOF of their tensor product is additive. 
\end{lemma}

Furthermore, recall that the entanglement cost is the regularized form of $E_f(\r)$, i.e., $E_c(\r)=\lim_{n\ra\infty}{1\over n}E_f(\r^{\otimes n})$.  So the state whose range is an EB space has equal entanglement cost and EOF. It provides a systematic method of deriving the additivity of EOF. 
	
Next, We shall refer to an $n$-partite biseparable state $\r$ on $\cH_{A_1A_2...A_n}$ as the convex sum of bipartite pure product states over the bipartition of systems $A_1,A_2,...,A_n$. In particular if the states are $n$-partite product states then we say that $\r$ is a fully separable state. If $n=3$ then we denote $A_1=A,A_2=B$ and $A_3=C$ for convenience. So a tripartite 
biseparable state on $\cH_{ABC}$ can be written as 
\begin{eqnarray}
\label{eq:bisep}
&&
\sum_i \proj{\lambda_i}_A\otimes\proj{\psi_i}_{BC}
+
\sum_j \proj{\mu_j}_B\otimes\proj{\phi_j}_{AC}
\notag\\+&&
\sum_k \proj{\nu_k}_C\otimes\proj{\o_k}_{AB}.
\end{eqnarray}
We express the state in \eqref{eq:bisep} as $\d_{ABC}+\e_{ABC}+\z_{ABC}$ where
\begin{eqnarray}
\label{eq:definite}
&&
\d_{ABC}:=\sum_i \proj{\lambda_i}_A\otimes\proj{\psi_i}_{BC},
\notag\\&&
\e_{ABC}:=\sum_j \proj{\mu_j}_B\otimes\proj{\phi_j}_{AC},
\notag\\&&
\z_{ABC}:=\sum_k \proj{\nu_k}_C\otimes\proj{\o_k}_{AB}.
\end{eqnarray}
We shall investigate the three states in the proof of Lemma \ref{le:rank2}. If a multipartite state is not biseparable then we say it is \textit{genuinely entangled (GE)}. To construct more GE states using known GE states, we review two ways of tensor product of two states in the Hilbert spaces $\cH_{A_1A_2\cdots A_n}$ and $\cH_{B_1B_2\cdots B_m}$ \cite{Shen_2020}. The first product is the tensor product $\cH_{A_1A_2\cdots A_n}\ox\cH_{B_1B_2\cdots B_m}$. Denote by $\rho\ox\sigma$ an $(n+m)$-partite state supported on the space $\cH_{A_1A_2\cdots A_n}\ox\cH_{B_1B_2\cdots B_m}$.
The second tensor product is the Kronecker product defined as
\begin{eqnarray}
\label{eq:krtenspace}
\cH_{A_1A_2\cdots A_n}\otimes_{K}\cH_{B_1B_2\cdots B_m}:=\big(\otimes_{i=1}^m (\cH_{A_i}\otimes \cH_{B_i})\big)\otimes \big(\otimes_{i'=m+1}^{n} \cH_{A_{i'}}\big),
\end{eqnarray}
where $m\le n$. 
We denote by $\rho\ox_K \sigma$ the state supported on the Hilbert space in \eqref{eq:krtenspace}. Hence $\r\otimes_K\sigma$ is an $n$-partite state of the systems $(A_1\ox B_1),\cdots,(A_m\ox B_m), A_{m+1},\cdots,A_n$. Now the following observations are clear.

\begin{lemma}
\label{le:rdmge}
(i) $\a\ox_K \b$ is an $n$-partite genuinely entangled state if $\a$ is $n$-partite genuinely entangled.

(ii) Suppose $\b$ is an $m$-partite fully separable state. Then $\a\ox_K \b$ is an $n$-partite genuinely entangled (resp. biseparable, fully separable) state if and only if $\a$ is an $n$-partite genuinely entangled (resp. biseparable, fully separable) state.
\end{lemma}

We point out that Conjecture \ref{cj:aotimesb} is related to the definition of $\a\otimes\b$, by defining that $\a_{AC_1}\otimes\b_{BC_2}$ is a tripartite state of systems $A,B$ and $C$, where $C=C_1C_2$. In the following, ref. \cite{Shen_2020}  has introduced a special Kronecker product. 

\begin{lemma}
	\label{le:kc}
Suppose $\a_{AC_{1,1}C_{1,2}\cdots C_{1,n}}$ and $\b_{BC_{2,1}C_{2,2}\cdots C_{2,n}}$ are two $(n+1)$-partite states on the Hilbert space $\mathcal{H}_{AC_{1,1}C_{1,2}\cdots C_{1,n}}$ and $\mathcal{H}_{BC_{2,1}C_{2,2}\cdots C_{2,n}}$. To
construct an $(n+2)$-partite state,  we apply the Kronecker product as follows,
\begin{eqnarray}
\mathcal{H}_{AC_{1,1}C_{1,2}\cdots C_{1,n}}\otimes_{K_c}\mathcal{H}_{BC_{2,1}C_{2,2}\cdots C_{2,n}}:=\mathcal{H}_A\otimes\mathcal{H}_B\otimes(\mathcal{H}_{C_{1,1}C_{1,2}\cdots C_{1,n}}\otimes_K\mathcal{H}_{C_{2,1}C_{2,2}\cdots C_{2,n}}). 
\end{eqnarray}
Denote by $\a\otimes_{K_c}\b$ a state supported on the space $\cH_{AC_{1,1}C_{1,2}\cdots C_{1,n}}\otimes_{K_c} \cH_{BC_{2,1}C_{2,2}···C_{2,n}}.$ Specially, if $n=1$, then $\a_{AC_1}\otimes\b_{BC_2}$ in Conjecture \ref{cj:aotimesb} can be written as $\a_{AC_1}\otimes_{K_c}\b_{BC_2}$.
\end{lemma}

Based on Lemma \ref{le:kc}, the following fact is proven in \cite{Shen_2020}. It is used in the proof of Theorem \ref{thm:main}.

\begin{lemma}
\label{le:2r2oxr2}
Conjecture \ref{cj:aotimesb} holds if one of the following conditions holds.

(i) One of $\a$ and $\b$ has rank one.

(ii) $\a$ and $\b$ both have rank two.

(iii) One of $\cR(\a)$ and $\cR(\b)$ is not spanned by product vectors.
\end{lemma}

Moreover, we recall the definition of distillable states \cite{DivincenzoEvidence}. A bipartite state $\r$ is n-distillable under LOCC if there exists a Schmidt-rank-two bipartite state $\ket{\psi}\in\mathcal{H}^{\otimes n}$ such that $\bra{\psi}(\r^{\otimes n})^\Gamma\ket{\psi}<0$. In quantum information theory, positive partial transpose (PPT) states are not distillable under LOCC \cite{Djokovic2016On}. So 
distillable states must be non-PPT (NPT) states \cite{Djokovic2016On}. We know that NPT states can be convert into NPT Werner states by LOCC \cite{Kraus2002Characterization,Reinaldo2006Distillability}. The Werner state is closely related to the distillability problem which lies in the heart of quantum entanglement theory \cite{D1999Separability}. The following is a well-known lemma on the distillability, and it is used in the proof of Lemma \ref{le:rank2}.
\begin{lemma}
	\label{le:werner}
	The Werner state $\r_w(d,p)$ is 
	
	(i) separable when $p\in[-\frac{1}{d},1]$;
	
	(ii) NPT and one-copy undistillable when $p\in[-\frac{1}{2},-\frac{1}{d})$;
	
	(iii) NPT and one-copy distillable when $p\in[-1,-\frac{1}{2})$.
\end{lemma}


The following fact is from \cite{Chen2007Entanglement}. It will be used to construct a tripartite state whose bipartite reduced density operators have additive EOF in Sec. \ref{sec:app}. This is an application of our results. 

\begin{lemma}
\label{le:addeofOTIMESsep=addeof}
Suppose $\a_{AB}$ has additive EOF and $\b_{CD}$ is a separable state. Then $\r_{AC:BD}=\a_{AB}\otimes\b_{CD}$ is a bipartite entangled state of additive EOF.	
\end{lemma}

\section{Tensor product of entanglement-breaking subspaces}
\label{sec:ebOXeb}

In this section, we investigate Conjecture \ref{cj:aotimesb} when the range of $\a$ and $\b$ are both EB subspaces. We say that a bipartite state $\r$ is locally projected onto (or locally convertible to) another state $\s$ when there exists a local operator $P\otimes Q$ such that $(P\otimes Q)\r(P^\dg\otimes Q^\dg)=\s$. We refer to the maximally correlated (MC) states as the states $\sum_{i,j}c_{ij}\ketbra{ii}{jj}$. One can show that the range of any MC state is an EB subspace. We present the main result of this section as follows.
\begin{theorem}
\label{thm:main}
Suppose $\a$ and $\b$ are two bipartite entangled states. Then Conjecture \ref{cj:aotimesb} holds when one of the following two conditions (i) and (ii) is satisfied.

(i)  $\a$ and $\b$ can be both locally projected onto entangled states of rank one or two.

(ii) $\cR(\a)$ and $\cR(\b)$ are  subspaces of two EB spaces spanned by $\{\ket{a_1,1},...,\ket{a_n,n}\}$ and $\{\ket{b_1,1},...,\ket{b_m,m}\}$, respectively.

(iii) Furthermore, if $\cR(\a)$ is a subspace of the EB space spanned by $\{\ket{a_1,1},...,\ket{a_n,n}\}$, then Conjecture \ref{cj:aotimesb} holds for $\a$ and all $\b$ if and only if it holds for every bipartite state whose range is spanned by $\ket{1,1}$ and $\ket{2,2}$ and all $\b$.
\end{theorem}
\begin{proof}
(i) According to Lemma \ref{le:2r2oxr2} (i), (ii), then we have (i) holds.

(ii) Let
$\a=\sum_{j\ge1} \proj{\ps_j}$ and $\ket{\ps_j}=\sum^n_{i=1} c_{ij}\ket{a_i,i}$. Using Wootters' decomposition, we may assume that $c_{11}\ne0$ and $c_{1j}=0$ for $j>1$. We have two cases (ii.a) and (ii.b). In case (ii.a), we assume that there exists $k$ such that $\ket{a_1}$ and $c_{k1}\ket{a_k}$ are linearly independent. One can show that the projected state $P \a P$ with $P=I\otimes (\proj{1}+\proj{k})$ is an entangled state of rank at most two. In case (ii.b), we consider the case that $\ket{a_1}$ and $c_{k1}\ket{a_k}$ are linearly dependent for any $k$. We can find a local projector $Q=I\otimes Q_1$ such that $Q\ket{\ps_1}=\ket{a_1,1}$ and $Q\ket{\ps_j}=\ket{\ps_j}$ for $j>1$. Hence $Q\a Q=\proj{a_1,1}+\sum_{j>1} \proj{\ps_j}$. Now we repeat the above argument to the state $\sum_{j>1} \proj{\ps_j}$, then we have $c_{22}\neq0$ and $c_{2j}=0$. We can find a local projector $R_1=I\otimes R_{11}$ such that $R_1\ket{\psi_1}=\ket{a_1,1}$, $R_1\ket{\psi_2}=\ket{a_2,2}$ and $R_1\ket{\psi_j}=\ket{j}, j>2$. Hence $R_1Q\a QR_1=\ketbra{a_1,1}{a_1,1}+\ketbra{a_2,2}{a_2,2}+\sum_{j>2}\ketbra{\psi_j}{\psi_j}$. Finally we will find a few local projectors $R_1,...,R_s$ such that $R_s...R_1 Q\a Q R_1...R_s=\sum_{i=1}^s\ketbra{a_i,i}{a_i,i}+\ketbra{\psi_{s+1}}{\psi_{s+1}}$, where $\ket{\psi_{s+1}}=\sum_{i=s+1}^nc_{i,s+1}\ket{a_i,i}$. Because $\a$ is an entangled state, then the state  $\ketbra{\psi_{s+1}}{\psi_{s+1}}$ is entangled. So  
we can project $R_s...R_1 Q\a Q R_1...R_s$ onto state $\ketbra{\psi_{s+1}}{\psi_{s+1}}$ 
with a projected operator $P'$, where
\begin{eqnarray}
\label{eq:proj}
P'=I\otimes(\ketbra{m_1}{m_1}+\ketbra{m_2}{m_2}), m_1,m_2\in\{s+1,\cdots,n\}.
\end{eqnarray}  
To conclude, we can always project $\a$ onto an entangled state of rank at most two. One can repeat the above argument to $\b$, and project $\b$ onto another entangled state of rank at most two. Now the claim follows from (i).

(iii) The "only if" part holds when $\mathcal{R}(\a)=\mathcal{R}(\g)$. We prove the "if" part. Up to local equivalence, suppose $\mathcal{R}(\g)$ is spanned by $\ket{a_1,1}$ and $\ket{a_2,2}$, then Conjecture \ref{cj:aotimesb} holds for $\g$ and all $\b$. Because $\mathcal{R}(\a)$ is a subspace of EB space spanned by $\{\ket{a_1,1},\cdots,\ket{a_n,n}\}$, then from (ii), there are two cases (iii.a) and (iii.b). In case (iii.a), we may assume that $\ket{a_1}$ and $c_{21}\ket{a_2}$ are linearly independent. we can show that the projected state $P\a P$ with $P=I\otimes(\ketbra{1}{1}+\ketbra{2}{2})$ is an entangled state in $\mathcal{R}(\g)$. In case (iii.b), we consider that $\ket{a_1}$ and $c_{k1}\ket{a_k}$ are linearly dependent for any $k$. From (ii.b), we can project $R_s...R_1 Q\a Q R_1...R_s$ onto state $\ketbra{\psi_{s+1}}{\psi_{s+1}}$ with projected operator $P'$ of \eqref{eq:proj}. Then we can find $Q'=I\otimes Q_2$ such that $Q'\ket{a_{m_1},m_1}=\ket{a_{m_1},1}$ and $Q'\ket{a_{m_2},m_2}=\ket{a_{m_2},2}$. So the  entangled state  $Q'P'R_s...R_1 Q\a Q R_1...R_sP'Q'$ is in $\mathcal{R}(\g)$.  According to the assumption, $\g_{AC_1}\otimes\b_{BC_2}\in\mathcal{H}_{ABC}$ is a tripartite GE state with $C=C_1C_2$. From the above cases, we can project the entangled state $\a$ in $\g$. Hence, $\a_{AC_1}\otimes\b_{BC_2}\in\mathcal{H}_{ABC}$ is also a tripartite GE state. 
\end{proof}

Note that the key of above proof is to show that the state whose range is an EB subspace can be converted to another state whose range is a $2$-dimensional EB subspace. Further, since the GE state remains a GE state up to local invertible operations, Theorem \ref{thm:main} (ii) and (iii) are valid if we replace $\ket{j}$'s by an arbitrary set of basis. 

Next, we shall investigate whether the condition in Theorem \ref{thm:main} (iii) holds. More explicitly, we investigate Conjecture \ref{cj:aotimesb} when the range of $\a$ is the simplest EB space, namely the two-qubit space spanned by $\ket{0,0}$ and $\ket{1,1}$. For this purpose, we present four facts with $\a$ is a bipartite entangled state of rank two in the following Lemma \ref{le:rank2}. By assuming $\a_{AC_1}\otimes_{K_c}\b_{BC_2}$ is a tripartite biseparable state, we  show specific representations and internal equivalences of $\a_{AC_1}, \d_{ABC}$ and  $\e_{ABC}$ defined in \eqref{eq:definite}. Then we show the proof of Lemma \ref{le:rank2} in Appendix \ref{app:le:bisep}.

\begin{lemma}
\label{le:rank2}
Suppose $\a$ is a bipartite entangled state of rank two, and $\b$ is a bipartite entangled state. If $\a_{AC_1}\otimes_{K_c} \b_{BC_2}$ is a tripartite biseparable state then

(i) $\a_{AC_1}\otimes_{K_c} \b_{BC_2}=
\d_{ABC}+\e_{ABC}$ defined in \eqref{eq:definite}.

(ii) up to local equivalence on systems $A$ and $C_1$ we may assume that
\begin{eqnarray}
\label{eq:sigmaABC}
\a_{AC_1}=&&
\cos^2\t(\cos\m\ket{0,0}+\sin\m\ket{1,1})(\cos\m\bra{0,0}+\sin\m\bra{1,1})
\notag\\+&&
\sin^2\t\proj{0,0},
\quad\quad \t,\m\in(0,\p/2),
\\\notag\\\label{eq:deltaABC}
\d_{ABC}=
&&
f\cos^2\n\proj{0,0}_{AC_1}\otimes (\b_0)_{BC_2}
+f\sin^2\n\proj{1,1}_{AC_1}\otimes (\b_1)_{BC_2},
\notag\\
f\in && (0,1),
\quad\quad
\n \in [0,\p/2],
\\\label{eq:epsilonABC}
\e_{ABC}=&&
(1-f)\sum^d_{j=1}
p_j
\proj{w_j}_B
\otimes
(
\cos\xi_j
\ket{0,0}_{AC_1}\ket{x_j}_{C_2}
+
\sin\xi_j
\ket{1,1}_{AC_1}\ket{y_j}_{C_2}
)
\notag\\&&
(\cos\xi_j
\bra{0,0}_{AC_1}\bra{x_j}_{C_2}
+\sin\xi_j
\bra{1,1}_{AC_1}\bra{y_j}_{C_2}
),
\notag\\&&
\sum^d_{j=1} p_j=1,
\quad
p_j>0,
\quad
\xi_j\in(0,\p/2),
\end{eqnarray}
where we have removed the two ends $\x_j=0$ and $\p/2$ by merging $\d_{ABC}$ with the pure states satisfying one of the two ends in $\e_{ABC}$.


Hence
\begin{eqnarray}
\label{eq:00}
(\cos^2\t\cos^2\m+\sin^2\t)\b_{BC_2} 	
=&&
f\cos^2\n(\b_0)_{BC_2}
+
(1-f)\sum^d_{j=1}p_j\cos^2\x_j\proj{w_j,x_j}_{BC_2},
\notag\\\\\label{eq:11}
\cos^2\t\sin^2\m\b_{BC_2} 	
=&&
f\sin^2\n(\b_1)_{BC_2}
+
(1-f)\sum^d_{j=1}p_j\sin^2\x_j\proj{w_j,y_j}_{BC_2},
\notag\\\\\label{eq:0011}
\cos^2\t\cos\m\sin\m\b_{BC_2} 	
=&&
(1-f)\sum^d_{j=1}p_j\cos\x_j\sin\x_j\ket{w_j,x_j}\bra{w_j,y_j}_{BC_2}.
\end{eqnarray}

(iii) if $\e_{ABC}^{\G_{AC_1}}=\e_{ABC}$, then $\e_{ABC}$ is a bipartite separable state with respect to the partition $AC_1$ and $BC_2$, i.e.,
\begin{eqnarray}
\label{eq:ac1:bc2}
\e_{ABC}=&&
(1-f)
\sum^r_{k=1}
q_k
(\cos\eta_k\ket{0,0}+\sin\eta_k\ket{1,1})
(\cos\eta_k\bra{0,0}+\sin\eta_k\bra{1,1})_{AC_1}
\otimes
\proj{\ps_k}_{BC_2},	
\notag\\&&
\sum^r_{k=1}
q_k=1,
\quad
q_k>0,
\quad
\eta_k\in(-\p/2,0)\cup(0,\p/2),
\notag\\&&
\cR(\b_{BC_2})=\lin\{\ket{\ps_k}\},
\quad
r=\rank\e_{ABC}\ge \dim \cR(\b_{BC_2}),
\end{eqnarray}
where we have removed the three points $\eta_k=-\p/2,0$ and $\p/2$ by the upcoming \eqref{eq:wjxj} and \eqref{eq:wjyj}.

By comparing \eqref{eq:epsilonABC} and \eqref{eq:ac1:bc2}, we have
\begin{eqnarray}
\label{eq:wjxj}
&&
\sum^d_{j=1}p_j\cos^2\x_j\proj{w_j,x_j}=\sum^r_{k=1}q_k\cos^2\eta_k\proj{\ps_k},
\\\label{eq:wjyj}
&&
\sum^d_{j=1}p_j\sin^2\x_j\proj{w_j,y_j}=\sum^r_{k=1}q_k\sin^2\eta_k\proj{\ps_k},
\\\label{eq:wjzj}
&&
\sum^d_{j=1}p_j\cos\x_j\sin\x_j\ket{w_j,x_j}\bra{w_j,y_j}=\sum^r_{k=1}q_k\cos\eta_k\sin\eta_k\proj{\ps_k}.
\end{eqnarray}


(iv) we may assume that $\b_{BC_2}$ is the Werner state $\r_w(d,\e-{1\over d})$ with some $\e\in[h,0)$ when $\b_{BC_2}$ is an NPT state.


\end{lemma}

In the next section, we shall apply our results to construct multipartite states having bipartite reduced density operators with additive EOF and whose bipartition is distillable under LOCC.

\section{Applications}
\label{sec:app}

In this section, we apply our results of previous sections to two constructions of multipartite states. First,  we construct a family of multipartite GE states having bipartite reduced density operators with additive EOF. Second, ref.  \cite{Chen2012NONDISTILLABLE} has investigated the distillability of three bipartite reduced density operators from a tripartite pure state. Then we further manage to construct a family of multipartite GE states whose every bipartition is a distillable state under LOCC. 
Such multipartite states are of widely usefulness in quantum-information tasks \cite{Chen2011Multicopy,Chen2012NONDISTILLABLE}.  We present Theorem \ref{thm:EOF:dis} as the main result in this section. 

\begin{theorem}
	\label{thm:EOF:dis}
 We present two constructions (i), (ii) of multipartite GE states.
	
(i) Suppose that $\a_{A_1C_1}^{(1)},\cdots,\a_{A_nC_n}^{(n)}$ are entangled states whose range are EB subspaces. We construct an $(n+1)$-partite state $\r_{A_1\cdots A_nC}=\a_{A_1C_1}^{(1)}\otimes\cdots\otimes\a_{A_nC_n}^{(n)}$, where $C=C_1\cdots C_n$. 

(i.a) For $n=2$,  $\r_{{A_1A_2C}}$ is a tripartite GE state whose three bipartite reduced density operators have additive EOF .

(i.b) When the $(n+1)$-partite state $\r_{{A_1\cdots A_nC}}$ is GE, every bipartite reduced density operator of $\r_{{A_1\cdots A_nC}}$ has additive EOF.    

(i.c) The EOF of every bipartition of the tripartite GE state $\r_{{A_1A_2C}}$ is additive.

(ii) Suppose that $\a_{A_1C_1}^{(1)},\cdots,\a_{A_nC_n}^{(n)}$ are entangled MC states. We construct an $(n+1)$-partite state $\r_{A_1\cdots A_nC}=\a_{A_1C_1}^{(1)}\otimes\cdots\otimes\a_{A_nC_n}^{(n)}$, where $C=C_1\cdots C_n$. 

(ii.a)   For $n=2$, $\r_{{A_1A_2C}}$ is a tripartite GE state whose every bipartition is distillable.

(ii.b)  When the $(n+1)$-partite state $\r_{{A_1\cdots A_nC}}$ is GE, every bipartition of $\r_{{A_1\cdots A_nC}}$ is distillable.   
\end{theorem}
\begin{proof}
	(i.a) We construct a family of tripartite GE states, any one of whose three bipartite reduced density operators have additive EOF. As far as we know, this is the first example of such GE states. Let $\r_{A_1A_2C}=\a_{A_1C_1}^{(1)}\otimes_{K_c}\a_{A_2C_2}^{(2)}$, where $\a_{A_1C_1}^{(1)}$ and $\a_{A_2C_2}^{(2)}$ are two entangled states whose ranges are subspaces of EB spaces spanned by $\{\ket{a_1,1},...,\ket{a_n,n}\}$ and $\{\ket{b_1,1},...,\ket{b_m,m}\}$, respectively. It follows from Theorem \ref{thm:main} that $\r_{A_1A_2C}$ is a tripartite GE state. Since $\r_{A_1A_2}$ is a separable, its EOF is zero and additive. Next, Lemma \ref{le:eb=additive=eof}	 implies that $\a_{A_1C_1}^{(1)}$ and $\a_{A_2C_2}^{(2)}$ both have additive EOF. 
	Because $\r_{A_1C}=\a_{A_1C_1}^{(1)}\otimes\a_{C_2}^{(2)}$ and $\r_{A_2C}=\a_{C_1}^{(1)}\otimes\a_{A_2C_2}^{(2)}$, they both have additive EOF in terms of Lemma \ref{le:addeofOTIMESsep=addeof}. We have finish the construction. It means that every bipartite reduced density operator of $\r_{A_1A_2C}$ has equal entanglement cost and EOF. In particular if the EOF is computable, then we can work out the entanglement cost of $\r_{A_1C}$, $\r_{A_2C}$ and $\r_{A_1A_2}$. This case occurs when $\a_{A_1C_1}^{(1)}$ and $\a_{A_2C_2}^{(2)}$ are two-qubit states using the known formula by Wootters \cite{wootters1998}.

(i.b) In Fig. \ref{fig:geeof}, we extend the forementioned tripartite GE states to $(n+1)$-partite GE states. 
Suppose that $\a_{A_1C_1}^{(1)}, \cdots, \a_{A_nC_n}^{(n)}$ are entangled states whose ranges are EB subspaces. We construct the $(n+1)$-partite state $\r_{A_1\cdots A_nC}=\a_{A_1C_1}^{(1)}\otimes_{K_c}\cdots\otimes_{K_c}\a_{A_nC_n}^{(n)}$, where $C=C_1\cdots C_n$. 
We obtain that bipartite reduced density operators of $\r_{A_1\cdots A_nC}$ are $\r_{A_pA_q}$ and $\r_{A_lC}$, where $p,q,l\in\{1,\cdots,n\}$ and $p< q$. One can verify that  $\r_{A_pA_q}=\a_{A_p}^{(p)}\otimes\a_{A_q}^{(q)}$. Since $\r_{A_pA_q}$ is a separable state, its EOF is zero and additive. On the other hand one can verify that $\r_{A_lC}=\a_{A_lC_l}^{(l)}\otimes_{K_c}(\otimes_{j\in\mathcal{M}}\a_{C_j}^{(j)})$ by tracing out remaining systems of $\r_{A_1\cdots A_nC}$, where $l\in\{1,\cdots,n\}$ and $\mathcal{M}=\{1,\cdots,n\}\setminus\{l\}$. Then we have $(\otimes_{j\in\mathcal{M}}\a_{C_j}^{(j)})$ is a separable 
state. One can show that EOF is additive for separable states. Because $\a_{A_lC_l}^{(l)}$ and $(\otimes_{j\in\mathcal{M}}\a_{C_j}^{(j)})$ have both additive EOF, Lemma \ref{le:eof} implies that $\r_{A_lC}$ also has additive EOF.  Hence, every bipartite reduced density operator of $\r_{A_1\cdots A_nC}$ has equal entanglement cost and EOF. 


\begin{figure}[htb]
	\includegraphics[scale=0.6,angle=0]{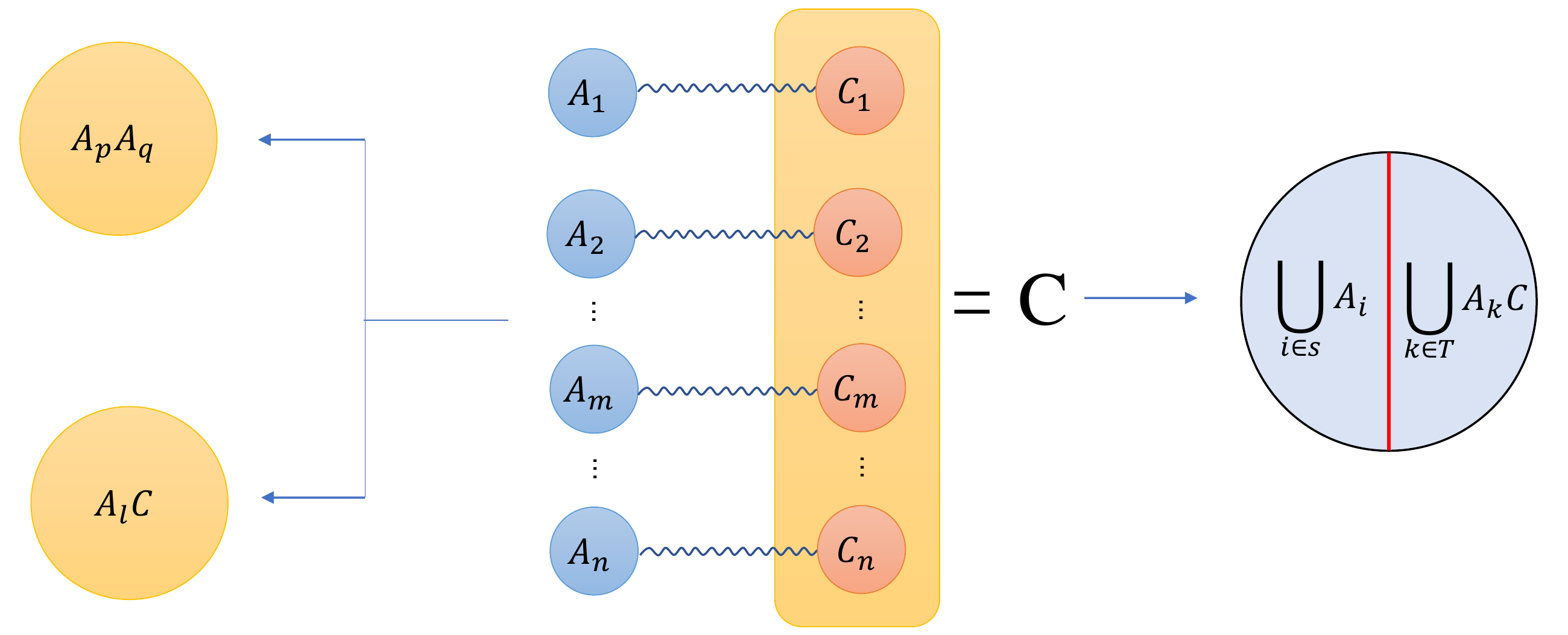}
	\caption{In the middle of this figure, every $\a_{A_jC_j}^{(j)}$ is an entangled state whose range is an EB subspace. Suppose an $(n+1)$-partite state $\r_{A_1\cdots A_nC}:={\a}_{A_1C_1}^{(1)}\otimes_{K_c}\cdots\otimes_{K_c}{\a}_{A_nC_n}^{(n)}$, where $C=C_1\cdots C_n$. One can verify that $\a_{A_jC_j}^{(j)}$'s all have additive EOF from Lemma \ref{le:eb=additive=eof}. The state $\r_{A_1\cdots A_nC}$ has bipartite reduce density operators $\r_{A_pA_q}$ and $\r_{A_lC}$ on the left side of this figure, where $p,q,l\in\{1,\cdots,n\}$ and $p<q$. Then every bipartite reduce density operator has additive EOF. On the other hand, we may regard that $\r_{A_1\cdots A_nC}$ is a bipartite state of the two systems $\cup_{i\in\mathcal{S}}A_i$ and $\cup_{k\in\mathcal{T}}A_kC$ on the right side, where $\mathcal{S}\cup\mathcal{T}=\{1,\cdots,n\}$ and $\mathcal{S}\cap\mathcal{T}=\emptyset$.  We obtain that every bipartite reduced density operator of the $(n+1)$-partite GE state $\r_{A_1\cdots A_nC}$ has additive EOF.} 
	\label{fig:geeof}
\end{figure}

(i.c) Suppose that two bipartite entangled states $\a_{A_1C_1}^{(1)}$ and $\a_{A_2C_2}^{(2)}$ both have additive EOF. We construct a tripartite state $\r_{A_1A_2C}=\a_{A_1C_1}^{(1)}\otimes_{K_c}\a_{A_2C_2}^{(2)}$, where $C=C_1C_2$. Then we regard $\r_{A_1A_2C}$ as a bipartite state $\r_{A_1A_2:C}$ in  $\mathcal{B}(\mathcal{H}_{A_1A_2}\otimes\mathcal{H}_{C})$. 
From Lemma \ref{le:eof} and the definition of additive EOF, we obtain that %
the EOF of $\r_{{A_1A_2:C}}$ is additive. Furthermore, if we regard $\r_{A_1A_2C}$ as a bipartite state $\r_{{A_1:A_2C}}$, then we obtain that
\begin{eqnarray}
\label{eq:eof:a:bc}
E_f(\r_{A_1:A_2C}\otimes\s_{BD})
&=&
E_f(\a_{A_1C_1}^{(1)}\otimes\s_{BD})\nonumber\\
&=&
E_f(\a_{A_1C_1}^{(1)})+E_f({\s_{BD}})\nonumber\\
&=&
E_f(\r_{A_1:A_2C})+E_f({\s_{BD}}).
\end{eqnarray}

So the bipartite state $\r_{A_1:A_2C}$ has additive EOF. Similarly, the bipartite state $\r_{A_2:A_1C}$ has additive EOF. Hence, from Theorem \ref{thm:main}, we obtain that $\r_{A_1A_2C}$ is a tripartite GE state whose every bipartition has additive EOF.

(ii.a) The distillable entanglement of a multipartite mixed state heavily evaluates the usefulness of this state to quantum computing and teleportation. Recently, Ref. \cite{1809.04202} established the bidistillable subspace in which every multipartite state is distillable in terms of any bipartition of this state, by constructing the so-called unextendible biseparable bases. Here we present another construction of multipartite state whose bipartition is distillable.   
We begin by studying the tripartite system. We recall that the entangled maximally correlated (MC) state is distillable under LOCC \cite{Horodecki2007Quantum}. Further, one can verify that the range of MC states is an EB subspace spanned by $\{\ket{a_1,1},...,\ket{a_n,n}\}$. Suppose $\a_{AC_1}$ and $\b_{BC_2}$ are two entangled MC states, and $\r_{ABC}=\a_{AC_1}\otimes_{K_c}\b_{BC_2}$. Then $\r_{AC}$ and $\r_{BC}$ are both distillable. We obtain that $\r_{ABC}$ is a tripartite GE state such that any bipartition of this state is distillable.

(ii.b) In Fig. \ref{fig:gemc}, we can extend the example to $(n+1)$-partite GE states whose every bipartition generates a distillable state. Suppose that $\a_{A_1C_1}^{(1)},\cdots,\a_{A_nC_n}^{(n)}$ are entangled MC states, and an $(n+1)$-partite state $\r_{A_1\cdots A_nC}:=\a_{A_1C_1}^{(1)}\otimes_{K_c}\cdots\otimes_{K_c}\a_{A_nC_n}^{(n)}$, where $C=C_1\cdots C_n$. Because $\a_{A_jC_{j}}^{(j)}$'s are entangled MC states, it implies that $\a_{A_jC_{j}}^{(j)}$'s are distillable under LOCC. Then one can asymptotically distill pure entangled states $\ket{\psi_j}_{A_jC_j}$'s from $\a_{A_jC_j}^{(j)}$'s, respectively. Then it follows that $\ket{\psi_1}_{A_1C_1}\otimes_{K_c}\cdots\otimes_{K_c}\ket{\psi_n}_{A_nC_n}$ is an $(n+1)$-partite GE state. Otherwise, if $\ket{\psi_1}_{A_1C_1}\otimes_{K_c}\cdots\otimes_{K_c}\ket{\psi_n}_{A_nC_n}$ is not GE, then we obtain that its density matrix $\r'_{A_1\cdots A_nC}$ is not GE.  
We may assume that $\r'_{A_1\cdots A_nC}$ is a bipartite separable state on systems $\cup_iA_i$ and $\cup_kA_kC$, where $i\in\mathcal{S}, k\in\mathcal{T}$ and $\mathcal{S}\cup\mathcal{T}=\{1,\cdots,n\}, \mathcal{S}\cap\mathcal{T}=\emptyset$. Denote $m$ as the maximum element in $\mathcal{S}$. Then $\r'_{A_m|C_m}$ is a bipartite separable state by tracing out remaining systems of $\r'_{A_1\cdots A_nC}$ and $\r'_{A_m|C_m}\propto\proj{\psi_s}_{A_mC_m}$. It is a contradition with the fact that $\ket{\psi_m}_{A_mC_m}$ is a pure entangled state. So $\r'_{A_1\cdots A_nC}$ is a bipartite pure entangled state in terms of every bipartition. Hence, every bipartition is distillable. We obtain that $\r_{A_1\cdots A_nC}$ is an $(n+1)$-partite GE state such that any bipartition of this state is distillable. 

\begin{figure}[htb]
	\includegraphics[scale=0.7,angle=0]{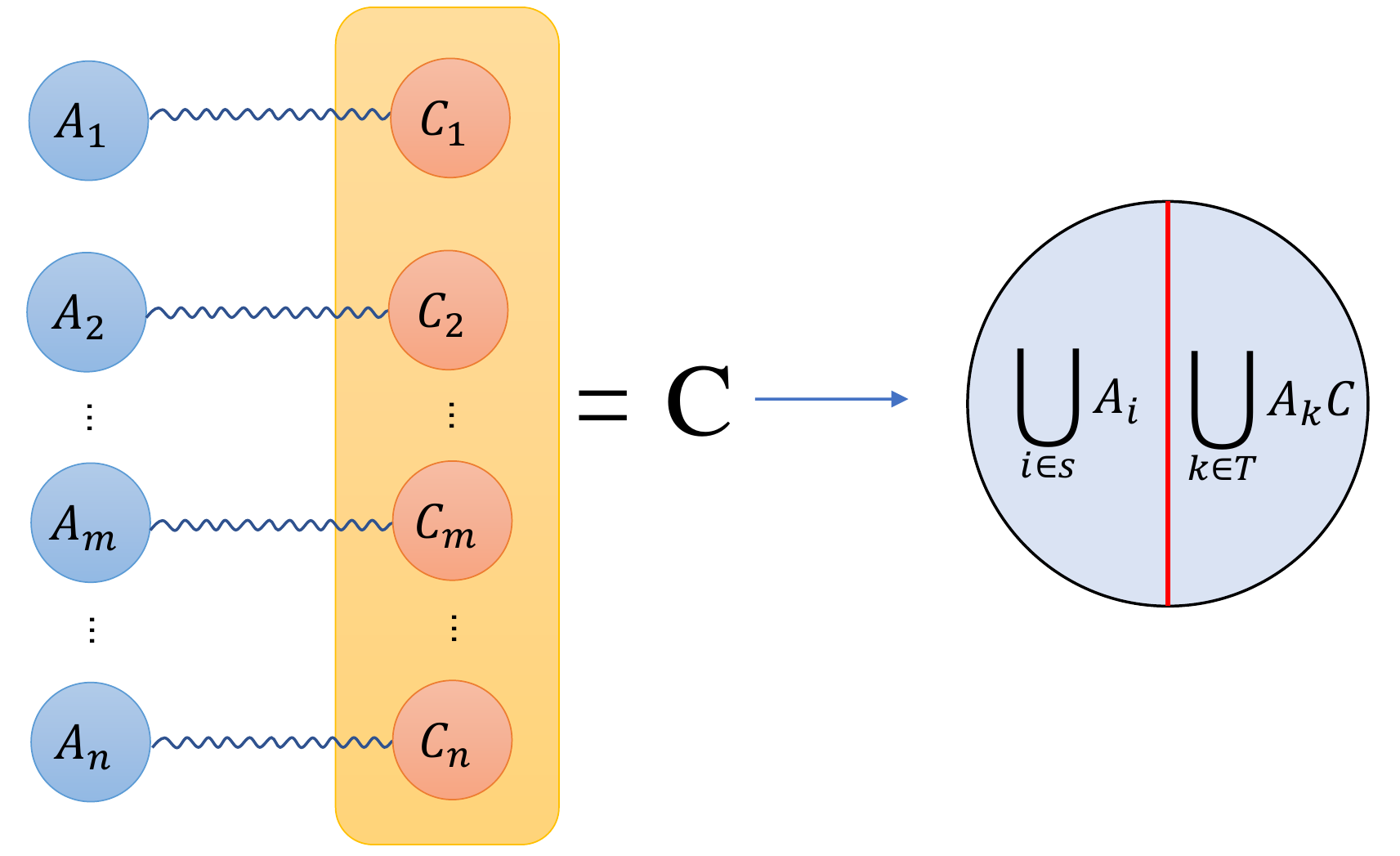}
	\caption{Every $\a_{A_jC_j}^{(j)}$ is an entangled MC state and it can be asymptotically distilled into a pure entangled state $\ket{\psi_j}_{A_jC_j}$, $1\leq j\leq n$. Suppose the $(n+1)$-partite state $\r'_{A_1\cdots A_nC}$ is the density matrix of $\ket{\psi_1}_{A_1C_1}\otimes_{K_c}\cdots\otimes_{K_c}\ket{\psi_n}_{A_nC_n}$. On the left side of this figure, suppose that $\r'_{A_1\cdots A_nC}$ is a GE state of the systems $A_1, \cdots, A_n, C$. We may regard that it is a bipartite state of the two systems $\cup_{i\in\mathcal{S}}A_i$ and $\cup_{k\in\mathcal{T}}A_kC$ on the right side of this figure, where $\mathcal{S}\cup\mathcal{T}=\{1,\cdots,n\}$ and $\mathcal{S}\cap\mathcal{T}=\emptyset$.  We obtain that $\a_{A_1C_1}\otimes_{K_c}\cdots\otimes_{K_c}\a_{A_nC_n}$ is an $(n+1)$-partite GE state such that any bipartition of this state is distillable.  
	} 
	\label{fig:gemc}
\end{figure}




\end{proof}

Moreover, there exists a relation between additive EOF and distillability. We may assume that a tripartite state $\r_{A_1A_2C}:=\a_{A_1C_1}^{(1)}\otimes_{K_c}\a_{A_2C_2}^{(2)}$, where  $\a_{A_1C_1}^{(1)},\a_{A_2C_2}^{(2)}$ are entangled MC states and $C=C_1C_2$. One can verify that the range of MC states is an EB subspace spanned by $\{\ket{a_1,1},\cdots,\ket{a_n,n}\}$. From Theorem \ref{thm:main}, the tripartite state $\r_{A_1A_2C}$ is GE. Above the two constructions of multipartite states in Theorem \ref{thm:EOF:dis}, (i.a) and (ii.a) show that any bipartition of $\r_{A_1A_2C}$ is distillable and has additive EOF.




\section{Multipartite genuine entanglement}
\label{sec:multiGE}

In this section, we explore more ways of constructing multipartite GE states by extending Conjecture \ref{cj:aotimesb}. We start by studying the tripartite case. In particular, we construct Conjecture \ref{cj:aotimesbotimesc=sep}, and explain its relations to Conjecture \ref{cj:aotimesb} in Lemma \ref{thm:bi->tri}. 


\begin{figure}[htb]
 	\includegraphics[scale=0.8,angle=0]{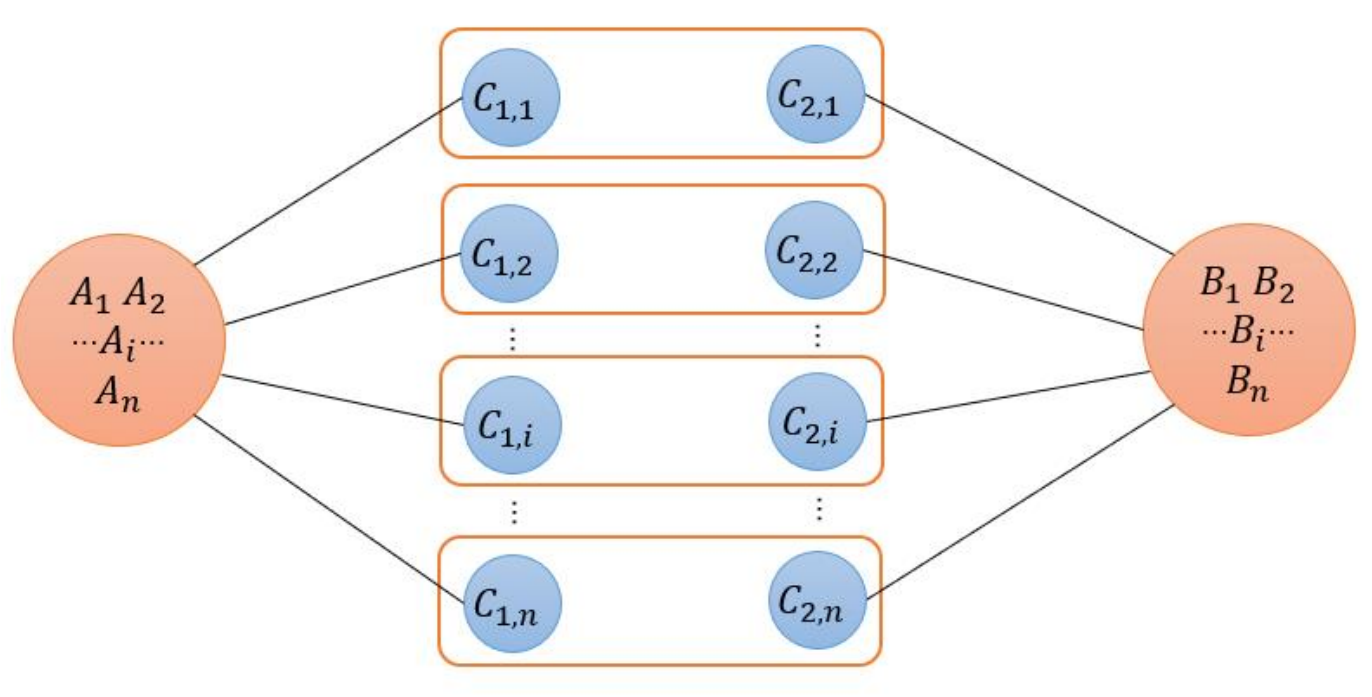}
 	\caption{If $\a_{AC_{1,1}C_{1,2}\cdots C_{1,n}}, \b_{BC_{2,1}C_{2,2}\cdots C_{2,n}}$ are two $(n+1)$-partite GE states on systems $A, C_{1,1},\cdots C_{1,n}$ and $B, C_{2,1},\cdots C_{2,n}$, then is the $(n+2)$-partite state  $\r_{ABC_1\cdots C_n}$  also a GE state, where $\r_{ABC_1\cdots C_n}=\a\otimes\b$,   $A=A_1A_2\cdots A_n$, $B=B_1B_2\cdots B_n$ and $C_j=C_{1,j}C_{2,j},1\leq j\leq n$?}
 	\label{fig:GMEn+2}
 \end{figure}
 
\begin{figure}[htb]	\includegraphics[scale=1,angle=0]{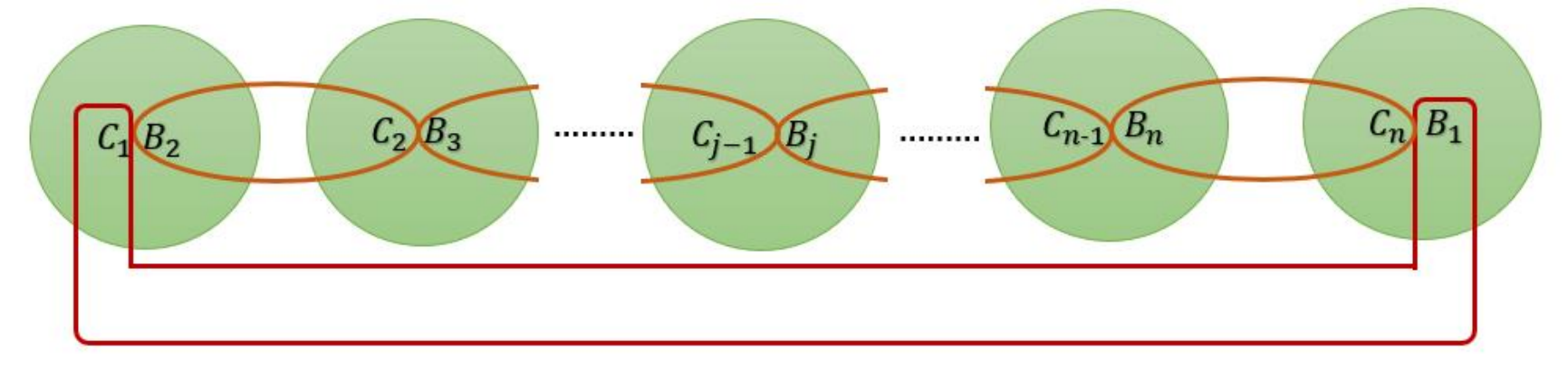}
	\caption{If $\a_{C_1B_2}^{(1)}, \cdots,\a_{C_{n}B_{1}}^{(n)}$ are bipartite entangled states, then is the $n$-partite state $\r_{A_1A_2\cdots A_n}$ also a GE state, where $\r_{A_1A_2\cdots A_n}=$ $\a_{C_1B_2}^{(1)}\otimes\cdots\otimes\a_{C_{n}B_{1}}^{(n)}$} and $A_j=B_jC_j$ ?
	\label{fig:binary2}
\end{figure}	

\begin{figure}[htb]
\includegraphics[scale=1.0,angle=0]{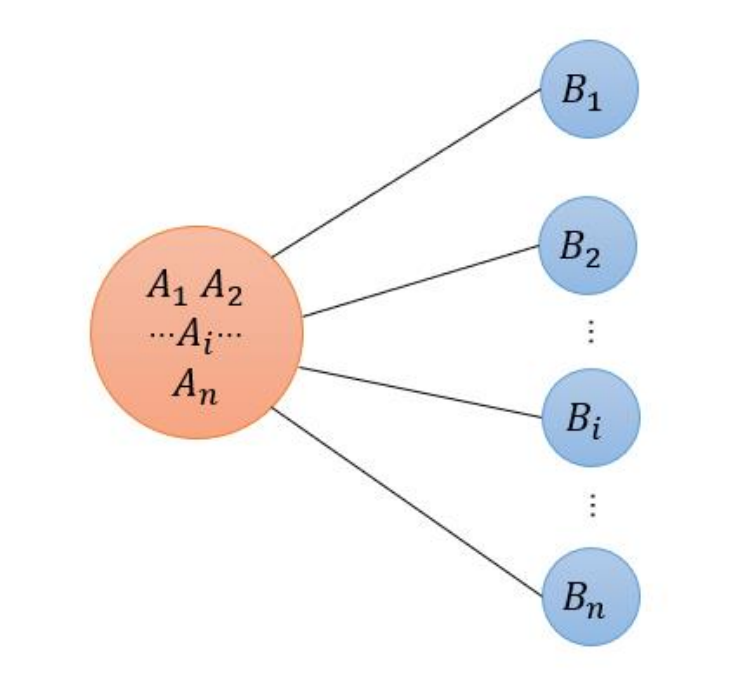}
	\caption{If $\a_{A_iB_i}^{(i)}$ is a GE state on systems $A_i$ and $B_i$, $i=1,\cdots,n$, then is the $(n+1)$-partite state  $\r_{AB_1\cdots B_n}$  also a GE state, where $\r_{AB_1\cdots B_n}=\a_{A_1B_1}^{(1)}\otimes\cdots\a_{A_nB_n}^{(n)}$ and  $A=A_1A_2\cdots A_n$? }
	\label{fig:GE123}
\end{figure}

\begin{conjecture}
\label{cj:aotimesbotimesc=sep}
We present five statements.

(i) In Fig. \ref{fig:GMEn+2}, we construct the $(n+2)$-partite state
\begin{eqnarray}
\label{eq:n+2}	
\r_{ABC_1...C_n}=\a_{AC_{1,1}C_{1,2}\cdots C_{1,n}}\otimes_{K_c} \b_{BC_{2,1}C_{2,2}\cdots C_{2,n}}
\end{eqnarray}
where $C_j=C_{1,j}C_{2,j}$ for $1\leq j\leq n$, and $\a_{AC_{1,1}C_{1,2}\cdots C_{1,n}}$ and $\b_{BC_{2,1}C_{2,2}\cdots C_{2,n}}$ are two   $(n+1)$-partite states. If they are both GE states, then $\r_{ABC_1...C_n}$ is GE.

(ii) Suppose $\a_{A_1B_1}$ and $\b_{A_2C_1}$ are two bipartite entangled states, $\g_{B_2C_2}$ is a bipartite state, and $\r_{ABC}=\a_{A_1B_1}\otimes\b_{A_2C_1}\otimes\g_{B_2C_2}$ is a tripartite state with $A=A_1A_2$, $B=B_1B_2$ and $C=C_1C_2$. If $\g$ is separable then $\r_{ABC}$ is GE.

(iii) Using the same notation in (ii), if $\g_{B_2C_2}$ is entangled then $\r_{ABC}$ is GE.

(iv) In Fig. \ref{fig:binary2}, we construct the $n$-partite state
\begin{eqnarray}
\label{eq:npartite3}
\r_{A_1...A_n}
=
\a^{(1)}_{C_1B_2}\otimes	
\a^{(2)}_{C_2B_3}\otimes	
\a^{(3)}_{C_3B_4}\otimes
...
\otimes
\a^{(n-1)}_{C_{n-1}B_n}\otimes	
\a^{(n)}_{C_nB_1},
\end{eqnarray}
where $A_j=B_jC_j$ and $\a^{(j)}$ is a bipartite state for $1\le j\le n$. If $\a^{(j)}$ is entangled, then $\r_{A_1...A_n}$ is GE.

(v) In Fig. \ref{fig:GE123}, if  
the $n$-partite state 
\begin{eqnarray}
\label{eq:ab1...bn}	
\r_{ABB_3\cdots B_n}=\a_{A_1B_1}^{(1)}\otimes\a_{A_2B_2}^{(2)}\otimes\cdots\otimes\a_{A_nB_n}^{(n)}
\end{eqnarray}
is GE, where $A=A_1A_2\cdots A_n$ and $B=B_1B_2$, then $\a_{A_jB_j}^{(j)}$ is entangled. Furthermore, if  $\a_{A_jB_j}^{(j)}$'s are entangled, then  $\r_{AB_1\cdots B_n}$ is also an $(n+1)$-partite GE state.   


\end{conjecture}

The motivation of this conjecture is to generate more multipartite GE states using bipartite states. We refer to Conjecture \ref{cj:aotimesbotimesc=sep} (v) as \textit{satellite mode generation of multipartite GE state} when it holds. In particular, one can show that the $n$-partite state $\r_{ABB_3\cdots B_n}$ is GE if the $(n+1)$-partite state $\r_{AB_1\cdots B_n}$ is GE   in Conjecture \ref{cj:aotimesbotimesc=sep} (v).  To characterize the relations between Conjecture \ref{cj:aotimesb} and \ref{cj:aotimesbotimesc=sep}, we show the following observation in Fig. \ref{fig:connections}. 

\begin{figure}[htb]
	\includegraphics[scale=1.0,angle=0]{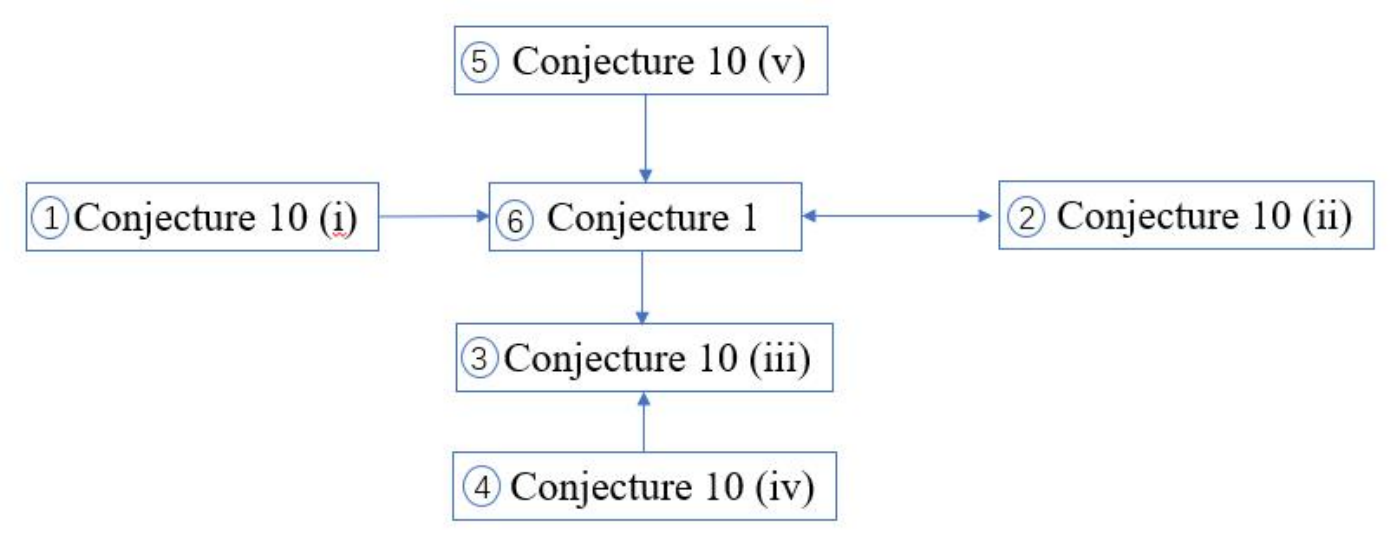}
	\caption{By investigating  the connections in Conjecture \ref{cj:aotimesbotimesc=sep}, we denote \textcircled{1}, $\cdots$, \textcircled{5} as Conjecture \ref{cj:aotimesbotimesc=sep} (i) $\cdots$, Conjecture \ref{cj:aotimesbotimesc=sep} (v) in order and \textcircled{6} as Conjecture \ref{cj:aotimesb}, respectively. There  exist three connections \textcircled{1} $\rightarrow$ \textcircled{6} and \textcircled{6} $\rightarrow$ \textcircled{3}, \textcircled{5} $\rightarrow$ \textcircled{6} and  \textcircled{6} $\leftrightarrow$  \textcircled{2} and \textcircled{4} $\rightarrow$  \textcircled{3}  at present. Then the connections are shown in Theorem \ref{thm:bi->tri}.}
	\label{fig:connections}
\end{figure}

\begin{lemma}
\label{thm:bi->tri}	

(i) Conjecture \ref{cj:aotimesbotimesc=sep} (i)  $\rightarrow$ Conjecture \ref{cj:aotimesb} $\rightarrow$ Conjecture \ref{cj:aotimesbotimesc=sep} (iii).

(ii) Conjecture \ref{cj:aotimesbotimesc=sep} (v) $\rightarrow$ Conjecture \ref{cj:aotimesb} $\leftrightarrow$  Conjecture \ref{cj:aotimesbotimesc=sep} (ii).

(iii) Conjecture \ref{cj:aotimesbotimesc=sep} (iv) $\rightarrow$ Conjecture \ref{cj:aotimesbotimesc=sep} (iii).








\end{lemma}
\begin{proof}
(i)   If $n=1$ then Conjecture \ref{cj:aotimesbotimesc=sep} (i) reduces to Conjecture \ref{cj:aotimesb}. We have proven the assertion. 

Suppose $\a_{A_1B_1}, \b_{A_2C_1}, \g_{B_2C_2}$ are three bipartite entangled states. Suppose 
$
\r_{ABC}
=
\a_{A_1B_1}\otimes\b_{A_2C_1}\otimes\g_{B_2C_2}
$, 
where $A=A_1A_2,B=B_1B_2$ and $C=C_1C_2$.  
Because  $\a_{A_1B_1}$ and $\b_{A_2C_1}$ are bipartite entangled states, when Conjecture \ref{cj:aotimesb} holds, we obtain that $\r'_{AB_1C_1}:=\a_{A_1B_1}\otimes \b_{A_2C_1}$ is a GE state. Then we rewrite $\r_{ABC}=\r'_{AB_1C_1}\otimes\g_{B_2C_2}$. If $\r_{ABC}$ is not a GE state, it implies that $\r_{ABC}$ is a biseparable state. We write $\r_{ABC}$ in \eqref{eq:bisep}. By tracing out systems $B_2, C_2$, we obtain  that $\r'_{AB_1C_1}$ is a tripartite biseparable state on systems $A, B_1$ and $C_1$. It is a contradiction with the fact that $\r'_{AB_1C_1}$ is a GE state. Hence, $\r_{ABC}$ is a GE state, and the Conjecture \ref{cj:aotimesbotimesc=sep} (iii) holds.

(ii)  If $n=2$ then the Conjecture \ref{cj:aotimesbotimesc=sep} (v) reduces to Conjecture \ref{cj:aotimesb}. Hence, all results for Conjecture \ref{cj:aotimesb} apply to Conjecture \ref{cj:aotimesbotimesc=sep} (v). Then we assume $\g_{B_2C_2}=\sum_i\ketbra{a_i,b_i}{a_i,b_i}$. Because $\g_{BC_2}$ is a bipartite separable state, we obtain 
\begin{eqnarray}
\label{eq:equ1}
\r_{ABC}
&=&
\a_{A_1B_1}\otimes\b_{A_2C_1}\otimes\g_{B_2C_2}\nonumber
\\
&=&
\sum_i(\a_{A_1B_1}\otimes\ketbra{a_i}{a_i}_{B_2}\otimes\b_{A_2C_1}\otimes\ketbra{b_i}{b_i}_{C_2})\nonumber
\\
&=&
\a'_{A_1B}\otimes\b'_{A_2C},
\end{eqnarray}	
where $\a'_{A_1B}$ and $\b'_{A_2C}$ are bipartite entangled states. 
From \eqref{eq:equ1}, then Conjecture \ref{cj:aotimesbotimesc=sep} (ii) is equivalent to Conjecture \ref{cj:aotimesb}. 

(iii) If $n=3$ then Conjecture \ref{cj:aotimesbotimesc=sep} (iv) reduces to Conjecture \ref{cj:aotimesbotimesc=sep} (iii). We have proven the assertion.

\end{proof}

\section{Conclusions}
	\label{sec:con}
We have investigated the tripartite state as the tensor product of two bipartite entangled states by merging two systems. We have shown that the tripartite state is a genuinely entangled state when the range of both bipartite states are entanglement-breaking subspaces. We also have investigated the tripartite state when the one of the two bipartite states has rank two. Further, we have constructed a family of multipartite GE states having bipartite reduced density operators with additive EOF, and another multipartite GE states whose every bipartition gives rise to a distillable state under LOCC. 
Moreover, we explored more ways of constructing multipartite genuinely entangled states. 

The next target is to investigate Conjecture \ref{cj:aotimesb} with the bipartite state whose range is an arbitrary EB space, e.g., the $2$-dimensional EB spaces in $\bbC^2\otimes\bbC^d$ constructed recently \cite{Zhao_2019}.  Further, an open problem is whether Conjecture \ref{cj:aotimesb} holds when the range of the bipartite state is a subspace of the EB space spaned by $\{\ket{a_i,i}\}_{i=1}^n$. Another direction is whether every bipartition of $(n+1)$-partite GE states has additive EOF. 

\section*{Acknowledgments}
	\label{sec:ack}	
Authors were supported by the  NNSF of China (Grant No. 11871089), and the Fundamental Research Funds for the Central Universities (Grant No. ZG216S2005).

\appendix

\section{Proof of Lemma \ref{le:rank2}}
\label{app:le:bisep}

By reviewing Lemma \ref{le:rank2}, we show the proof of Lemma \ref{le:rank2} from (i) to (iv).

\begin{proof}
	(i) From Lemma \ref{le:2r2oxr2}, we assume that $\a_{AC_1}$ is spanned by pure biseparable states. Since $\a_{AC_1}$ is a bipartite entangled state of rank two, then we assume that $\mathcal{R}(\a)$ is spanned by $\{\ket{a_1,b_1}, \ket{a_2,b_2}\}$. It implies that $\ket{a_1}, \ket{a_2}$ are linearly independent, and $\ket{b_1}, \ket{b_2}$ are linearly independent. Then we can find two invertible matrices $X,  Y$ such that
	\begin{eqnarray}
	\label{eq:equivalent}
	(X\otimes Y)\ket{a_1,b_1}=\ket{0,0},\quad (X\otimes Y)\ket{a_2,b_2}=\ket{1,1}.
	\end{eqnarray}
	So rank $\a=2$ implies that $\mathcal{R}(\a)$ is spanned by $\ket{0,0}$ and $\ket{1,1}$ up to local equivalence.  Furthermore $\cR(\d), \cR(\e), \cR(\z)\subseteq\a\otimes_{K_c}\b$. 
	Because $\mathcal{R}(\z_{AC_1})\subseteq\mathcal{R}(\a)$, by tracing out systems $B, C_2$, we obtain that $\z_{AC_1}$ is a separable state spanned by $\ket{0,0}$ and $\ket{1,1}$ in \eqref{eq:definite}. 
	From the definition of $\z_{ABC}$ in \eqref{eq:definite},  we can rewrite $\z_{ABC}=\sum_ic_i \ketbra{0,s_i,0,\psi_{i}}{0,s_i,0,\psi_{i}}_{AB(C_1C_2)}+\sum _jm_j\ketbra{1,v_j,1,w_j}{1,v_j,1,w_j}_{AB(C_1C_2)}$. 
	Then by tracing out systems $A,C_1$, we obtain that $\z_{BC_2}=\sum c_{i}\ketbra{s_i,\psi_i}{s_i,\psi_i}+\sum m_j\ketbra{v_j,w_j}{v_j,w_j}$, where $\ket{s_i,\psi_i},\ket{v_j,w_j}$ are also separable states. 
	Thus, $\z_{ABC}$ is a separable state on systems $A, B, C_1, C_2$. Then we can merge $\z_{ABC}$ with the part $\d_{ABC}$ or $\e_{ABC}$. Hence, $\a_{AC_1}\otimes_{K_c}\b_{BC_2}=\d_{ABC}+\e_{ABC}$.

	(ii) From Lemma \ref{le:2r2oxr2} and \eqref{eq:equivalent} in (i), up to local equivalence on systems $A$ and $C_1$, the condition rank $\a_{AC_1}=2$ implies that
	\begin{eqnarray}
	\label{eq:rangeAC1}
	\mathcal{R}(\a_{AC_1})=\lin\{\ket{0,0},\ket{1,1}\}.
	\end{eqnarray}
	Because $\a_{AC_1}$ is a bipartite entangled state of rank two, by normalizing $\a_{AC_1}$,  \eqref{eq:sigmaABC} holds.
	From \eqref{eq:definite}, we have
	\begin{eqnarray}
	\label{eq:daltaABC12}
	\d_{ABC}=\sum_i\ketbra{\l_i}{\l_i}_A\otimes\ketbra{\psi_i}{\psi_i}_{BC}.
	\end{eqnarray}
	By tracing out systems $B,C_2$,  we may assume that
	\begin{eqnarray}
	\label{eq:daltaAC1}
	\d_{AC_1}:=\tr_{BC_2}\d_{ABC}=\sum_i\proj{\l_i}_A\otimes(\eta_i)_{C_1}, 
	\end{eqnarray}
	where we denote $\tr_{BC_2}\ketbra{\psi_i}{\psi_i}_{BC}:=(\eta_i)_{C_1}$. 
	From the definition \eqref{eq:bisep}, we have $\cR(\d_{ABC})\subseteq \cR(\a_{AC_1}\otimes_{K_c}\b_{BC_2})$. We obtain that $\cR(\d_{AC_1})\subseteq\cR(\a_{AC_1})$ by tracing out systems $B, C_2$.  Eqs. \eqref{eq:rangeAC1} and \eqref{eq:daltaAC1} imply that the product vectors in $\cR(\delta_{AC_1})$ are $\ket{0,0}$ and $\ket{1,1}$. Then we have
	\begin{eqnarray}
	\label{eq:xy}
	\ket{x,y}:=a\ket{0,0}+b\ket{1,1}\in\cR(\d_{AC_1})
	\end{eqnarray}
	if and only if $ ab=0$, where $a,b$ are two complex numbers.
	From \eqref{eq:daltaAC1}, we have $\ket{\l_i}\otimes\cR(\eta_i)\subseteq\cR(\d_{AC_1})$. Then we have $\ket{\l_i}\otimes\ket{y}\in \cR(\d_{AC_1})$, where $\ket{y}\in\cR(\eta_i)$. If rank $\eta_i>1$,  then there exists $\ket{z}\in\cR(\eta_i)$ such that $\ket{z}$ is linearly independent with $\ket{y}$. Then we obtain that 
	$
	k_1\ket{\l_i, y}+k_2\ket{\l_i,z}\in \cR(\d_{AC_1}),
	$
	where $k_1,k_2$ are any complex numbers. So the range of $\delta_{AC_1}$ has infinitely many product vectors that are pairwise linearly independent. It is contradiction with \eqref{eq:xy}. So rank $\eta_i=1$. From the definition of $\d_{ABC}$ in \eqref{eq:definite} and rank $\eta_i=1$, 
	we obtain that $\ket{\psi_i}_{BC}$ is a bipartite separable state on systems $BC_2$ and $C_1$. Then we denote 
	\begin{eqnarray}
	\label{eq:psibc}
	\ket{\psi_i}_{BC}:=\ket{\phi_i}_{BC_2}\otimes \ket{b_i}_{C_1}.
	\end{eqnarray}
	
	Substituting \eqref{eq:psibc} into \eqref{eq:daltaABC12}, 
	we obtain that  
	$
	\d_{ABC}
	=
	\sum_i\proj{\l_i,b_i}_{AC_1}\otimes\proj{\phi_i}_{BC_2}.
	$
	From \eqref{eq:xy}, the $\ket{\l_i,b_i}$'s are proportional to product vectors $\ket{0,0}$ or $\ket{1,1}$. We rewrite  
	$
	\d_{ABC}=\sum_i\ketbra{0,0}{0,0}_{AC_1}\otimes\ketbra{\phi'_i}{\phi'_i}_{BC_2}+\sum_j\ketbra{1,1}{1,1}_{AC_1}\otimes\ketbra{\ph_j}{\ph_j}_{BC_2},
	$
	where $\ket{\phi'_i}, \ket{\ph_j}$ are proportional to the elements in $\ket{\phi_i}$'s.
	Let $\b_0=\sum_i\ketbra{\phi'_i}{\phi'_i}$ and  $\b_1=\sum_j\ketbra{\ph_j}{\ph_j}$. Then we have
	$
	\d_{ABC}=\ketbra{0,0}{0,0}_{AC_1}\otimes(\b_0)_{BC_2}+\ketbra{1,1}{1,1}_{AC_1}\otimes(\b_1)_{BC_2}.
	$
	By normalizing $\d_{ABC}$, we have 
	\begin{eqnarray}
	\label{eq:dalta2}
	\d_{ABC}=\cos^2\nu\ketbra{0,0}{0,0}_{AC_1}\otimes(\b_0)_{BC_2}+\sin^2\nu\ketbra{1,1}{1,1}_{AC_1}\otimes(\b_1)_{BC_2}, 
	\end{eqnarray}
	where $\nu\in[0,\pi/2]$. If $\tr\a_{AC_1}\otimes\b_{BC_2}=1$, from (i), there exists $f\in(0,1)$ such that
	\begin{eqnarray}
	\label{eq:f}
	\a_{AC_1}\otimes_{K_c}\b_{BC_2}=f\d_{ABC}+(1-f)\e_{ABC}.
	\end{eqnarray}
	So we rewrite \eqref{eq:dalta2} 
	as 
	\eqref{eq:deltaABC}. 
	
	Next, from \eqref{eq:definite}, we have 
	\begin{eqnarray}
	\label{eq:epsilon2}
	\e_{ABC}=\sum_{j}\proj{\mu_j}_B\otimes\proj{\phi_j}_{AC}.
	\end{eqnarray}
	
	In \eqref{eq:epsilon2}, 
	we may assume that $\ket{\phi_j}_{AC}$ is a bipartite pure state on systems $AC_1$ and $C_2$. Because $\mathcal{R}(\e_{ABC})\subseteq\cR(\a_{AC_1}\otimes_{K_c}\b_{BC_2})$, we have $\cR(\e_{AC_1})\subseteq\cR(\a_{AC_1})$. By using Schmidt decompostion, we have 
	\begin{eqnarray}
	\label{eq:mj}
	\ket{\phi_j}_{AC}=a_j\ket{0,0}_{AC_1}\ket{x_j}_{C_2}+b_j\ket{1,1}_{AC_1}\ket{y_j}_{C_2},
	\end{eqnarray}
	where $a_j,b_j$ are complex numbers and $\ket{x_j}, \ket{y_j}$ are states on system $C_2$. By normalizing \eqref{eq:epsilon2} and \eqref{eq:mj},  we rewrite \eqref{eq:epsilon2} as follows,
	\begin{eqnarray}
	\label{eq:epABC}
	\e_{ABC}&=&\sum_j p_j\proj{w_j}_B\otimes(\cos\xi_j\ket{0,0}_{AC_1}\ket{x_j}_{C_2}\nonumber
	\\
	&+&
	\sin\xi_j\ket{1,1}_{AC_1}\ket{y_j}_{C_2})
	(\cos\xi_j\bra{0,0}_{AC_1}\bra{x_j}_{C_2}+\sin\xi_j\bra{1,1}_{AC_1}\bra{y_j}_{C_2}),
	\end{eqnarray}
	where $\sum_j p_j=1, \xi_j\in[0,\pi/2]$. If $\x_j=0$ or $\pi/2$, then we can merge the states in $\d_{ABC}$. So we remove the two ends $\x_j=0$ and $\p/2$.
	From \eqref{eq:f}, we rewrite  \eqref{eq:epABC} as 
	\eqref{eq:epsilonABC}.
	Comparing \eqref{eq:sigmaABC}, \eqref{eq:deltaABC} and \eqref{eq:epsilonABC},  we obtain that \eqref{eq:00}, \eqref{eq:11}, \eqref{eq:0011} hold.
	
	(iii) Suppose that
	\begin{eqnarray}
	\label{eq:eabcabc}
	\e_{ABC}=\sum_{k=1}^rq_k\proj{\varphi_k}, \quad q_k>0,\quad \sum_k q_k=1,
	\end{eqnarray}
	and $r$ is the rank of $\e_{ABC}$. Because $\cR(\e_{ABC})\subseteq\cR(\a_{AC_1}\otimes_{K_c}\b_{BC_2})$, we have $\cR(\e_{AC_1})\subseteq\cR(\a_{AC_1})$. Eq. \eqref{eq:rangeAC1} implies that  $\cR(\e_{AC_1})=\lin\{\ket{0,0},\ket{1,1}\}$. By using Schmidt decomposition on systems $A,C_1$ and $B,C_2$, we may assume 
	\begin{eqnarray}
	\label{eq:akbk}
	\ket{\varphi_k}:=a_k\ket{0,0}_{AC_1}\ket{\psi_k}_{BC_2}+b_k\ket{1,1}_{AC_1}\ket{\psi'_k}_{BC_2},
	\end{eqnarray}
	where $a_k,b_k$ are complex numbers and $\ket{\psi_k},\ket{\psi'_k}$ are unit vectors. Substituting \eqref{eq:akbk} into \eqref{eq:eabcabc}, we have
	\begin{eqnarray}
	\label{eq:eabca}
	\e_{ABC}&=&\sum_{k=1}^rq_k(a_k\ket{0,0}_{AC_1}\ket{\psi_k}_{BC_2}+b_k\ket{1,1}_{AC_1}\ket{\psi'_k}_{BC_2})(a_k^*\bra{0,0}_{AC_1}\bra{\psi_k}_{BC_2}\nonumber
	\\
	&+&
	b^*_k\bra{1,1}_{AC_1}\bra{\psi'_k}_{BC_2}).
	\end{eqnarray}
	By normalizing \eqref{eq:eabca}, we rewrite 
	\begin{eqnarray}
	\label{eq:eabcunit}
	\e_{ABC}&=&\sum_{k=1}^rq_k(\cos\xi_k\ket{0,0}_{AC_1}\ket{\psi_k}_{BC_2}+\sin\xi_k\ket{1,1}_{AC_1}\ket{\psi'_k}_{BC_2})(\cos\xi_k\bra{0,0}_{AC_1}\bra{\psi_k}_{BC_2}\nonumber
	\\
	&+&
	\sin\xi_k\bra{1,1}_{AC_1}\bra{\psi'_k}_{BC_2}),
	\end{eqnarray}
	where $\xi_k\in[0,\pi/2]$.
	Then we have 
	\begin{eqnarray}
	\label{eq:abceplison}
	\e_{ABC}^{\Gamma_{AC_1}}&=&\sum_{k=1}^rq_k(\cos^2\xi_k\proj{0,0}_{AC_1}\otimes\proj{\psi_k}_{BC_2}\nonumber\\
	&+&\cos\xi_k\sin\xi_k\ketbra{1,1}{0,0}_{AC_1}\otimes\ketbra{\psi_k}{\psi'_k}_{BC_2}\nonumber
	\\
	&+&
	\sin\xi_k\cos\xi_k\ketbra{0,0}{1,1}_{AC_1}\otimes\ketbra{\psi'_k}{\psi_k}_{BC_2}\nonumber
	\\
	&+&
	\sin^2\xi_k\ketbra{1,1}{1,1}_{AC_1}\otimes\proj{\psi'_k}_{BC_2}
	).
	\end{eqnarray}
	
	If $\e_{ABC}=\e_{ABC}^{\Gamma_{AC_1}}$, by comparing \eqref{eq:eabcunit} and \eqref{eq:abceplison}, then we have 
	\begin{eqnarray}
	\label{eq:psipsi}
	\ketbra{\psi_k}{\psi'_k}_{BC_2}=\ketbra{\psi'_k}{\psi_k}_{BC_2}.
	\end{eqnarray}
	Eq. \eqref{eq:psipsi} implies that $\ket{\psi_k}_{BC_2}$ is proportional to $\ket{\psi'_k}_{BC_2}$. We may assume that 
	\begin{eqnarray}
	\label{eq:eabcmk}
	\ket{\psi'_k}:=m_k\ket{\psi_k},
	\end{eqnarray}
	where $m_k$ is a complex number. Substituting \eqref{eq:eabcmk} into \eqref{eq:abceplison}, we have
	\begin{eqnarray}
	\label{eq:eplisonabc3}
	\e_{ABC}^{\Gamma_{AC_1}}
	&=&
	\sum_{k=1}^rq_k(m_k\cos\xi_k\ket{0,0}+\sin\xi_k\ket{1,1})(m_k\cos\xi_k\bra{0,0}\nonumber\\
	&+&\sin\xi_k\bra{1,1})\otimes\proj{\psi_k}_{BC_2}.
	\end{eqnarray}
	By normalizing \eqref{eq:eplisonabc3}, we have 
	\begin{eqnarray}
	\label{eq:abcabcd}
	\e_{ABC}^{\Gamma_{AC_1}}
	&=&
	\sum_{k=1}^rq_k(\cos\eta_k\ket{0,0}+\sin\eta_k\ket{1,1})(\cos\eta_k\bra{0,0}\nonumber\\
	&+&\sin\eta_k\bra{1,1})\otimes\proj{\psi_k}_{BC_2},
	\end{eqnarray}
	where $\eta_k\in[-\pi/2,\pi/2]$. If $\eta_k=0, \pi/2$ or $-\pi/2$, then we can merge the states in $\d_{ABC}$. So we remove $\eta_k=0, \pi/2$ and $-\pi/2$. From \eqref{eq:f}, we rewrite \eqref{eq:abcabcd} as \eqref{eq:ac1:bc2}. By comparing \eqref{eq:epsilonABC} and \eqref{eq:ac1:bc2}, we obtain that \eqref{eq:wjxj}, \eqref{eq:wjyj}, \eqref{eq:wjzj} hold.

	By tracing out systems $A,C_1$ in \eqref{eq:ac1:bc2}, we have $\e_{BC_2}=(1-f)\sum_{k}q_k\proj{\psi_k}_{BC_2}$. By definition $\cR(\e_{BC_2})=\lin\{\ket{\psi_k}\}_{k=1}^r$. Because $\cR(\e_{ABC})\subseteq\a_{AC_1}\otimes_{K_c}\b_{BC_2}$, we have $\cR(\e_{BC_2})\subseteq\cR(\b_{BC_2})$. If the inclusion is proper, then there exists a nonzero vector $\ket{a}\in\cR(\b_{BC_2})$ and $\ket{a}\perp\cR(\e_{BC_2})$. Meanwhile, $\ket{a}$ is not orthogonal to $\d_{BC_2}$. Then we obtain that 
	\begin{eqnarray}
	\a_{AC_1}\propto\bra{a}(\a_{AC_1}\otimes_{K_c}\b_{BC_2})\ket{a}
	\end{eqnarray}
	is a bipartite separable state. It is contradiction with the fact that $\a_{AC_1}$ is entangled. So we have $\cR(\b_{BC_2})=\cR(\e_{BC_2})=\lin\{\ket{\psi_k}\}_{k=1}^r$. 
	From \eqref{eq:ac1:bc2}, we have $r=\rank \e_{ABC}\geq\dim\cR(\e_{BC_2})$. So we obtain that $r=\rank \e_{ABC}\geq\dim\cR(\b_{BC_2})$.
	
	
	(iv) It's known that each NPT bipartite state can be convert into an NPT Werner state by using LOCC \cite{DiVincenzo_2000,D_r_2000}. From Lemma \ref{le:werner}, if $\b_{BC_2}\in\mathbb{C}^d\otimes\mathbb{C}^d$ is an NPT state, then there exists an LOCC operator $\Delta$ such that $\Delta(\b_{BC_2})=\r_w(d,p)$, where $p\in[-1,-\frac{1}{d})$. So we may assume that $\b_{BC_2}$ is the Werner state $\r_w(d,\e-\frac{1}{d})$ with some $\e\in[h,0)$, where $[h,0)$ is a neighborhood.
	
\end{proof}

\bibliographystyle{unsrt}

\bibliography{20200609_conjecture_on_AC1_otimes_BC2}

\end{document}